\documentclass[conference]{IEEEtran}

\usepackage[utf8]{inputenc}
\usepackage[T1]{fontenc}
\usepackage{url}
\usepackage{ifthen}
\usepackage{cite}
\usepackage[cmex10]{amsmath}

\usepackage{graphicx}
\usepackage{amsfonts}
\usepackage{amsmath}
\usepackage{amssymb}
\usepackage{mathrsfs} 
\usepackage{color}
\usepackage{bm}
\usepackage{latexsym} 

\newtheorem{thm}{Theorem}

\newtheorem{lem}{Lemma}
\newtheorem{proof}{proof}

\newtheorem{defn}{Definition}
\newtheorem{rem}{Remark}
\newtheorem{exam}{Example}

\begin{document}

\title{List-decodable Codes for Single-deletion Single-substitution
with List-size Two}

\author{\IEEEauthorblockN{Wentu~Song, Kui~Cai, and Tuan~Thanh~Nguyen}
\IEEEauthorblockA{Science, Mathematics and Technology Cluster\\
Singapore University of Technology and Design, Singapore 487372\\
Email: \{wentu\_song, cai\_kui, tuanthanh\_nguyen\}@sutd.edu.sg}}

\maketitle


\begin{abstract}
In this paper, we present an explicit construction of
list-decodable codes for single-deletion and single-substitution
with list size two and redundancy $3\log n+4$, where $n$ is the
block length of the code. Our construction has lower redundancy
than the best known explicit construction by Gabrys \emph{et al}.
$($arXiv 2021$)$, whose redundancy is $4\log n+O(1)$.
\end{abstract}

\section{Introduction}

Codes correcting insertion, deletion and substitution errors
(collectively referred to as edit errors) have gone through a long
history from the seminal work of Levenshtein \cite{Levenshtein65}.
It was shown in \cite{Levenshtein65} that the binary
Varshamov-Tenengolts (VT) code \cite{Varshamov65}, which is given
by \begin{align*} \mathscr C_{n}(a)=\left\{\bm x\in\{0,1\}^n:
\sum_{i=1}^nix_i\equiv a~(\text{mod}~n+1)\right\},\end{align*} can
correct a single edit error and is asymptotically optimal in
redundancy, given by $\log n+2$. Order-optimal non-binary
single-edit correcting codes were studied in
\cite{Cai19,Gabrys-21}.

Constructing optimal multiple-edit error correcting codes is much
more challenging, even for binary deletion codes. A generalization
of the VT construction for multiple-deletion correcting codes was
presented in \cite{Helberg02}, but this generalized construction
has asymptotic rate strictly smaller than $1$. Recently, there
were many works on explicit construction of low-redundancy
$t$-deletion correcting codes for $t\geq 2~($e.g., see
\cite{Brakensiek18}$-$\!\!\cite{Gur2020}$)$. For $t=2$, Guruswami
and H\aa stad constructed a family of $2$-deletion correcting
codes with length $n$ and redundancy $4\log n+O(\log\log n)$
\cite{Gur2020}, which matches the best known upper bound obtained
via the Gilbert-Varshamov-type greedy algorithm
\cite{Brakensiek18}. By introducing the higher order VT syndromes
and the syndrome compression technique, Sima \emph{et al}.
constructed a family of $t$-deletion correcting codes with
redundancy $8t\log n+o(\log n)$ \cite{Sima19-2-1}. Unfortunately,
for $t>2$, all existential constructions of $t$-deletion
correcting codes have redundancy greater than the
Gilbert-Varshamov-type bound.

The best known $t$-edit correcting codes for $t\geq 2$ were given
by Sima \emph{et al}., which have redundancy $4t\log n+o(\log n)$
\cite{Sima20}. The method in \cite{Sima20} was improved by the
authors in \cite{SNCH-21}, which gave a construction of
$t$-deletion $s$-substitution correcting codes with redundancy
$(4t+3s)\log n+o(\log n)$. A family of single-deletion
single-substitution correcting binary codes with redundancy $6\log
n+8$ was constructed in \cite{Smagloy20}. So far, constructing
optimal (with respect to redundancy) multiple-edit correcting
codes is still an open problem, even for single-deletion
single-substitution correcting codes.

As a relaxation of the decoding requirement, list-decoding for
insertions and deletions have been considered by several research
teams, mainly focusing on list-decoding for some fraction of
deletions/insertions \cite{Guruswami17}$-$\!\!\cite{Liu21}. Unlike
the traditional decoding (also referred to as unique-decoding),
list-decoding with list-size $\ell$ allows to give a set of $\ell$
codewords from each corrupted sequence. A family of explicit
list-decodable codes for two deletions with length $n$ and
list-size two was constructed in \cite{Gur2020}, which has
redundancy $3\log n$. Note that the redundancy of the construction
in \cite{Gur2020} is lower than the Gilbert-Varshamov-type bound,
which is $4\log n$. The improvement in redundancy is achieved by
the relaxation in the decoding requirement.

In this paper, we present an explicit construction of
list-decodable codes for single-deletion and single-substitution
with list-size two and redundancy $3\log n+4$. Our construction
improves the recent work by Gabrys \emph{et al}. \cite{Gabrys-21},
which constructed such codes with redundancy $4\log n+O(1)$.

The rest of this paper is organized as follows. In Section
\uppercase\expandafter{\romannumeral 2}, the basic concepts are
introduced and some preliminary properties of the errors are
discussed. Our construction of list-decodable codes for
single-deletion and single-substitution is presented in Section
\uppercase\expandafter{\romannumeral 3}. The auxiliary lemma used
by our construction is proved in Section
\uppercase\expandafter{\romannumeral 4}.

\section{Preliminaries}

For any positive integers $m$ and $n$ such that $m\leq n$, denote
$[m,n]=\{m,m+1,\ldots,n\}$. If $m>n$, let $[m,n]=\emptyset$. For
simplicity, we denote $[n]=[1,n]$ and $\mathbb Z_n=[0,n-1]$.

In this work, we consider binary codes. For any sequence (vector)
$\bm x$ of length $n$, we use $x_i$ to denote the $i$th symbol of
$\bm x$, and hence $\bm x$ can be denoted as $\bm
x=(x_1,x_2,\ldots, x_n)\in\{0,1\}^n$ or simply, $\bm
x=x_1x_2\ldots x_n$. The \emph{weight} of $\bm x$, denoted by
$\mathsf{wt}(\bm x)$, is the number of non-zero symbols (the
symbol $1$ for binary sequence) in $\bm x$. Clearly, for binary
sequence $\bm x$, we have $\mathsf{wt}(\bm x)=\sum_{i=1}^nx_i$.

Given non-negative integers $t$ and $s$ such that $t+s<n$, for any
$\bm{x}\in\{0,1\}^n$, the \emph{error ball} of $\bm x$ under
$t$-deletion $s$-substitution, denoted by $\mathscr
B_{t,s}(\bm{x})$, is the set of all sequences that can be obtained
from $\bm{x}$ by $t$ deletions $($i.e., deleting $t$ symbols of
$\bm x)$ and at most $s$ substitutions $($i.e., substituting at
most $s$ symbols of $\bm x$, each with a different symbol$)$. A
code $\mathscr C\subseteq\{0,1\}^n$ is \emph{list-decodable} for
$t$-deletion $s$-substitution with list size $\ell$ if any $\bm
y\in\{0,1\}^{n-1}$ is contained by the error ball of at most
$\ell$ codewords of $\mathscr C$. In other words, for any $\bm
y\in\{0,1\}^{n-1}$, there exist at most $\ell$ codewords of
$\mathscr C$ from which $\bm y$ can be obtained by $t$ deletions
and at most $s$ substitutions.

In this work, we consider list-decodable codes for single deletion
and single substitution, i.e., $t=s=1$. Suppose $\bm
x\in\{0,1\}^n$ and $\bm y\in\{0,1\}^{n-1}$ such that $\bm y$ can
be obtained from $\bm x$ by deleting one symbol of $\bm x$ and
substituting at most one symbol of $\bm x$ with a different symbol
in $\{0,1\}$. We can compute the difference between the weights of
$\bm x$ and $\bm y$ for all possible cases $($see Table 1$)$.
According to Table 1, we have $\mathsf{wt}(\bm x)-\mathsf{wt}(\bm
y)\in\{-1,0,1,2\}$. If $\mathsf{wt}(\bm x)-\mathsf{wt}(\bm
y)\in\{-1,2\}$, then the values of the deleted and substituted
symbols can be determined. If $\mathsf{wt}(\bm x)-\mathsf{wt}(\bm
y)=0$, then $\bm y$ can be obtained from $\bm x$ by deleting a
$0$, or by deleting a $1$ and substituting a $0$ with a $1$. For
the case that $\bm y$ is obtained from $\bm x$ by deleting a $0$,
unless $\bm x$ is the all-zero sequence, $\bm y$ can also be
obtained from $\bm x$ by deleting a $1$ and substituting a $0$
with a $1$.\footnote{For example, let $\bm x=0110001$ and let $\bm
y=011001$ be obtained from $\bm x$ by deleting $x_5=0$. Then $\bm
y$ can also be obtained from $\bm x$ by deleting $x_3=1$ and
substituting $x_4=0$ with $\bar{x}_4=1$. In general, if $\bm x$ is
not the all-zero sequence and $\bm y$ can be obtained from $\bm x$
by deleting a $0$, then we can always find a $0~($denoted by
$\hat{0})$ in the same run with the deleted $0$ that is adjacent
to a $1~($denoted by $\hat{1})$. Then $\bm y$ can always be viewed
as being obtained from $\bm x$ by deleting the $\hat{1}$ and
substituting the $\hat{0}$ with $1$.} Hence, if $\mathsf{wt}(\bm
x)-\mathsf{wt}(\bm y)=0$, then $\bm y$ can always be obtained from
$\bm x$ by deleting a $1$ and substituting a $0$ with a $1$.
Similarly, if $\mathsf{wt}(\bm x)-\mathsf{wt}(\bm y)=1$ and $\bm
x$ is not the all-one sequence, then $\bm y$ can always be
obtained from $\bm x$ by deleting a $0$ and substituting a $1$
with a $0$. In summary, we have the following remark.
\begin{rem}\label{rem-dw-xy}
Suppose $\bm x\in\{0,1\}^n\backslash\{1^n,0^n\}$, where $1^n$ and
$0^n$ are the all-one sequence and the all-zero sequence of length
$n$ respectively, and $\bm y\in\{0,1\}^{n-1}$ such that $\bm y$
can be obtained from $\bm x$ by deleting one symbol of $\bm x$ and
substituting \emph{at most} one symbol of $\bm x$. Then $\bm y$
can be obtained from $\bm x$ by deleting one symbol and
substituting \emph{exactly one} symbol of $\bm x$, and the values
of the deleted and substituted symbols can be determined by
$\mathsf{wt}(\bm x)\!~(\text{mod}~4)$ and $\mathsf{wt}(\bm y)$.
\end{rem}
\begin{table}
\begin{center}\renewcommand\arraystretch{1.3}
\begin{tabular}{|p{3.1cm}|p{2.0cm}|}
\hline Cases of error combination &
~$\mathsf{wt}(\bm x)-\mathsf{wt}(\bm y)$ \\
\hline ~~$1\rightarrow \epsilon$,~ no substitution & $~~~~~~~~~~1$ \\
\hline ~~$1\rightarrow \epsilon,~ 1\rightarrow 0$  & $~~~~~~~~~~2$ \\
\hline ~~$1\rightarrow \epsilon,~ 0\rightarrow 1$  & $~~~~~~~~~~0$ \\
\hline ~~$0\rightarrow \epsilon$,~ no substitution & $~~~~~~~~~~0$ \\
\hline ~~$0\rightarrow \epsilon,~ 1\rightarrow 0$  & $~~~~~~~~~~1$ \\
\hline ~~$0\rightarrow \epsilon,~ 0\rightarrow 1$  & $~~~~~~~-\!1$ \\
\hline
\end{tabular}\\
\end{center}
\vspace{1.5mm}Table 1. The value of $\mathsf{wt}(\bm
x)-\mathsf{wt}(\bm y)$ for different cases of single-deletion
single-substitution, where $a\rightarrow \epsilon$ means a symbol
$a\in\{0,1\}$ is deleted from $\bm x$ and for
$b\in\{0,1\}\backslash\{a\}$, $a\rightarrow b$ means a symbol $a$
of $\bm x$ is substituted by the symbol $b$.
\end{table}

\section{Main Results}

In this section, we present our construction of list-decodable
codes for single-deletion and single-substitution. Our
construction only uses the weight and the first two order VT
syndromes for binary sequences.

We adopt the method of \cite{Sima19-2-1} to define the higher
order VT syndromes. For each positive integer $j$ and each
$\bm{x}\in\{0,1\}^n$, the $j$th-order VT syndrome of $\bm{x}$ is
defined as
\begin{align}\label{def-f}
f_j(\bm{x})=\sum_{i=1}^n\left(\sum_{\ell=1}^i
\ell^{j-1}\right)x_i.
\end{align}
As in \cite{SNCH-20}, we can rearrange the terms and obtain
\begin{align}\label{re-def-f}
f_j(\bm{x})=\sum_{\ell=1}^n\left(\sum_{i=1}^\ell
i^{j-1}\right)x_\ell
=\sum_{i=1}^n\left(\sum_{\ell=i}^nx_\ell\right)i^{j-1}.
\end{align}

The code is given by the following definition, where $1^n$ and
$0^n$ denote the all-one sequence and the all-zero sequence of
length $n$, respectively.

\begin{defn}\label{code-Con}
For any fixed values $c_0\in\mathbb Z_4$, $c_1\in\mathbb Z_{2n}$
and $c_2\in\mathbb Z_{2n^2}$, let $\mathscr C_{n}(c_0,c_1,c_2)$ be
the set of all sequences $\bm x\in\{0,1\}^n\backslash\{1^n, 0^n\}$
satisfying the following three conditions:
\begin{enumerate}
 \item[$(\text{C}0)$] $\mathsf{wt}(\bm x)\equiv c_0~(\text{mod}~4)$.
 \item[$(\text{C}1)$] $f_1(\bm x)\equiv c_1~(\text{mod}~2n)$.
 \item[$(\text{C}2)$] $f_2(\bm x)\equiv c_2~(\text{mod}~2n^2)$.
\end{enumerate}
\end{defn}

Then our main result can be stated as the following theorem.
\begin{thm}\label{main-thm}
There exists a $(c_0,c_1,c_2)\in\mathbb Z_4\times\mathbb
Z_{2n}\times\mathbb Z_{2n^2}$ such that the code $\mathscr
C_{n}(c_0,c_1,c_2)$ in Definition \ref{code-Con} has redundancy at
most $3\log n+4$ and is list-decodable from single-deletion and
single-substitution with list size $2$.
\end{thm}

In the rest of this section, we always assume that
$(c_0,c_1,c_2)\in\mathbb Z_4\times\mathbb Z_{2n}\times\mathbb
Z_{2n^2}$ and $\mathscr C_{n}(c_0,c_1,c_2)$ is given by Definition
\ref{code-Con}. For any $\bm x\in\{0,1\}$ and
$\{d,e\}\subseteq[n]$, let $E(\bm x,d,e)$ denote the sequence
obtained from $\bm{x}$ by deleting $x_{d}$ and substituting
$x_{e}$ with $\bar{x}_e=1-x_e~($i.e., $\bar{x}_e=1$ if $x_e=0$ and
$\bar{x}_e=0$ if $x_e=1)$. Clearly, $E(\bm x,d,e)\in\{0,1\}^{n-1}$
is uniquely determined by $\bm x,d$ and $e$. We also need the
following lemma, which will be proved in Section IV.

\begin{lem}\label{main-lem}
Suppose $\bm{x}$, $\bm{x}'\in\mathscr C_{n}(c_0,c_1,c_2)$ and
$\{d_1,e_1\}$, $\{d_2,e_2\}\subseteq[n]$ such that
$\bm{x}\neq\bm{x}'$, $d_{1}\leq d_{2}$ and
$E(\bm{x},d_1,e_1)=E(\bm{x}',d_2,e_2)$. We have $d_1<e_1\leq d_2$
and $d_1\leq e_2<d_2$.
\end{lem}

In formally speaking, if there exists a $\bm y\in\{0,1\}^{n-1}$
such that $\bm y$ can be obtained from $\bm x$ and $\bm x'$ by
deleting one symbol and substituting one symbol, then the two
substituted symbols are both located between the two deleted
symbols.

\vspace{2pt}Using Lemma \ref{main-lem}, we can prove Theorem
\ref{main-thm} as follows.

\begin{proof}[Proof of Theorem \ref{main-thm}]
By the pigeonhole principle, there exists a
$(c_0,c_1,c_2)\in\mathbb Z_4\times\mathbb Z_{2n}\times\mathbb
Z_{2n^2}$ such that the code $\mathscr C_{n}(c_0,c_1,c_2)$ has
size at least $\frac{2^n-2}{16n^3}$, hence the redundancy of
$\mathscr C_{n}(c_0,c_1,c_2)$ is at most $3\log n+4$.

It remains to prove that $\mathscr C_{n}(c_0,c_1,c_2)$ is
list-decodable from single-deletion and single-substitution with
list size $2$. We need to prove that for any given
$\bm{y}\in\{0,1\}^{n-1}$, there exist at most two codewords in
$\mathscr C_{n}(c_0,c_1,c_2)$, from which $\bm{y}$ can be obtained
by one deletion and at most one substitution. This can be proved
by contradiction as follows.

Suppose $\bm{x},\bm{x}'$ and $\bm x''$ are three distinct
sequences in $\mathscr C_{n}(c_0,c_1,c_2)$ from which $\bm{y}$ can
be obtained by one deletion and at most one substitution. By
Remark \ref{rem-dw-xy}, we can assume $\bm
y=E(\bm{x},d_1,e_1)=E(\bm{x}',d_2,e_2)=E(\bm{x}'',d_3,e_3)$.
Without loss of generality, assume $d_1\leq d_2$.

First, consider $\bm x$ and $\bm x'$. By Lemma \ref{main-lem}, we
have
\begin{align}\label{eq-x-p1}
d_1<e_1\leq d_2\end{align} and \begin{align}\label{eq-x-p2}
d_1\leq e_2<d_2.\end{align}

For further discussions, we have the following three cases.

Case 1: $d_3\leq d_1$. Considering $\bm x$ and $\bm x''$, by Lemma
\ref{main-lem}, we have $d_3\leq e_1<d_1$ and $d_3<e_3\leq d_1$.
Combining with \eqref{eq-x-p1}, we have $e_1<d_1<e_1$, a
contradiction.

Case 2: $d_1<d_3\leq d_2$. Considering $\bm x$ and $\bm x''$, by
Lemma \ref{main-lem}, we have $d_1<e_1\leq d_3$ and $d_1\leq
e_3<d_3$. On the other hand, considering $\bm x'$ and $\bm x''$,
by Lemma \ref{main-lem}, we have $d_3<e_3\leq d_2$ and $d_3\leq
e_2<d_2$. Hence, we obtain $e_3<d_3<e_3$, a contradiction.

Case 3: $d_2<d_3$. Considering $\bm x'$ and $\bm x''$, by Lemma
\ref{main-lem}, we have $d_2<e_2\leq d_3$ and $d_2\leq e_3<d_3$.
Combining with \eqref{eq-x-p2}, we get $e_2<d_2<e_2$, a
contradiction.

From the above discussions, we can conclude that there exist at
most two codewords in $\mathscr C_{n}(c_0,c_1,c_2)$ from which
$\bm{y}$ can be obtained by one deletion and at most one
substitution, which proves Theorem \ref{main-thm}.
\end{proof}

\section{Proof of Lemma \ref{main-lem}}

In this section, we prove Lemma \ref{main-lem}. We always suppose
that $\bm{x}$, $\bm{x}'\in\mathscr C_{n}(c_0,c_1,c_2)$, and
$\{d_1,e_1\}$, $\{d_2,e_2\}\subseteq[n]$ such that $d_{1}\leq
d_{2}$ and $E(\bm{x},d_1,e_1)=E(\bm{x}',d_2,e_2)$. We first
enumerate all the possible cases according to the order of
$d_1,e_1,d_2,e_2$.

\begin{rem}\label{rem-cases}
Consider $e_1, d_1$ and $d_2$. Since $d_1\leq d_2$, we have three
cases: $e_1<d_1$, $d_1<e_1\leq d_2$ and $d_2<e_1$. Similarly, for
$e_2, d_1$ and $d_2$, we have three cases: $e_2<d_1$, $d_1\leq
e_2<d_2$ and $d_2<e_2$. Combining these two scenarios we have a
total of nine cases to consider. However, we can merge some cases
and consider the following six cases.
\begin{itemize}
 \item[(i)] $e_1<d_1\leq d_{2}$ and $e_2<d_1\leq d_{2}$.
 \item[(ii)] $e_1<d_1\leq e_2<d_{2}$ or $e_2<d_1<e_1\leq d_{2}$.
 \item[(iii)] $e_1<d_1\leq d_2<e_{2}$ or $e_2<d_1\leq d_{2}<e_1$.
 \item[(iv)] $d_1<e_1\leq d_2$ and $d_1\leq e_{2}<d_2$.
 \item[(v)] $d_1<e_1\leq d_2<e_{2}$ or $d_1\leq e_2<d_{2}<e_1$.
 \item[(vi)] $d_1\leq d_2<e_1$ and $d_1\leq d_{2}<e_2$.
\end{itemize}
\end{rem}

\vspace{2pt}We will prove that $\bm x=\bm x'$ for all cases in
Remark \ref{rem-cases} except for Case (iv). Hence, if $\bm
x\neq\bm x'$, then it must fall into Case (iv), that is,
$d_1<e_1\leq d_2$ and $d_1\leq e_2<d_2$.

Denote $\bm{x}=(x_1,x_2,\ldots,x_n)$ and
$\bm{x}'=(x'_1,x'_2,\ldots,x'_n)$. For each $i\in[n]$, let
\begin{align}\label{def-u-i}
u_i\triangleq\sum_{\ell=i}^nx_\ell-\sum_{\ell=i}^nx'_\ell.\end{align}
To prove $\bm{x}=\bm{x}'$, it suffices to prove $u_i=0$ for all
$i\in[n]$.

The following lemma will be used in our discussions. $($Recall
that for each positive integer $j$ and $\bm x\in\{0,1\}^n$,
$f_j(\bm x)$ is the $j$th-order VT syndrome of $\bm{x}$ defined by
\eqref{def-f} or \eqref{re-def-f}.$)$

\begin{lem}\label{PT-sgn-eq}
Let $m$ be a fixed positive integer. Suppose
$\big(f_1(\bm{x}),\ldots,f_{m+1}(\bm{x})\big)
=\big(f_1(\bm{x}'),\ldots,f_{m+1}(\bm{x}')\big)$ and there exist
$m$ positive integers, say $p_1,p_2,\ldots,p_{m}$, such that
$1\leq p_1<p_2<\cdots<p_{m}\leq n$ and for each $j\in[m+1]$,
either $u_i\geq 0$ for all $i\in[p_{j-1}+1,p_j]$ or $u_i\leq 0$
for all $i\in[p_{j-1}+1,p_j]$, where $p_0=1$ and $p_{m+1}=n$. Then
$u_i=0$ for all $i\in[n]$, and hence we have $\bm x=\bm x'$.
\end{lem}

The proof of Lemma \ref{PT-sgn-eq} is omitted because it is
(implicitly) contained in the proof of \cite[Proposition
2]{Sima19-2-1}.

The following simple remark is also useful in our proof.

\begin{rem}\label{rem-wt-xd}
Since $\mathsf{wt}(\bm x)\equiv\mathsf{wt}(\bm x')\equiv
c_0~(\text{mod}~4)~($because $\bm{x}$, $\bm{x}'\in\mathscr
C_{n}(c_0,c_1,c_2))$ and $E(\bm{x},d_1,e_1)=E(\bm{x}',d_2,e_2)$,
then by Remark \ref{rem-dw-xy}, we have $x_{d_1}=x'_{d_2}$ and
$\mathsf{wt}(\bm x)=\mathsf{wt}(\bm x')$.
\end{rem}

In the following five subsections, we will prove that for all
cases in Remark \ref{rem-cases} except for Case (iv), we have
$\big(f_1(\bm{x}),f_{2}(\bm{x})\big)
=\big(f_1(\bm{x}'),f_{2}(\bm{x}')\big)$, and there exists a
$p_1\in[n]$ such that for each $j\in\{1,2\}$, either $u_i\geq 0$
for all $i\in[p_{j-1}+1,p_j]$ or $u_i\leq 0$ for all
$i\in[p_{j-1}+1,p_j]$, where $p_0=1$ and $p_{2}=n$. Then by Lemma
\ref{PT-sgn-eq} $($for the special case of $m=1)$, we have $\bm
x=\bm x'$. Thus, if $\bm x\neq\bm x'$, then it must fall into Case
(iv), that is, $d_1<e_1\leq d_2$ and $d_1\leq e_2<d_2$.

\subsection{Proof of $\bm{x}=\bm{x}'$ for Case (i)}

For this case, we have $e_1<d_1\leq d_{2}$ and $e_2<d_1\leq
d_{2}$. If $e_1=e_2$, then $\bm x=\bm x'$.\footnote{If
$e_1=e_2<d_1\leq d_{2}$, then there is a $\bm y'\in\{0,1\}^{n-1}$
such that $\bm y'$ can be obtained from $\bm x~($resp. $\bm x')$
by a single deletion. By (C1) of Definition \ref{code-Con},
$\mathscr C_{n}(c_0,c_1,c_2)$ is a single-deletion correcting
code, so we can obtain $\bm x=\bm x'$.} Therefore, we assume
$e_1\neq e_2$. To simplify the presentation, let
$\lambda_1=\min\{e_1,e_2\}$ and $\lambda_2=\max\{e_1,e_2\}$. Then
$1\leq \lambda_1<\lambda_2<d_1\leq d_2$. Since
$E(\bm{x},d_1,e_1)=E(\bm{x}',d_2,e_2)$, we can obtain\footnote{In
fact, let $\bm y=E(\bm{x},d_1,e_1)=E(\bm{x}',d_2,e_2)$, which
means that $\bm y$ can be obtained from $\bm{x}$ by deleting
$x_{d_1}$ and substituting $x_{e_1}$ with
$\bar{x}_{e_1}=1-x_{e_1}$, and $\bm y$ can also be obtained from
$\bm{x}'$ by deleting $x'_{d_2}$ and substituting $x'_{e_2}$ with
$\bar{x}'_{e_2}=1-x'_{e_2}$. Then \eqref{xxp-case1} can be
obtained by comparing the elements of $\bm x$ and $\bm x'$ with
the elements of $\bm y$: $x_i=y_i=x'_i$ for
$i\in[1,d_1-1]\backslash \{\lambda_1, \lambda_2\}$;
$x_i=y_{i-1}=x'_{i-1}$ for $i\in[d_1+1,d_2]$; and $x_i=y_i=x'_i$
for $i\in[d_2+1,n]$.}
\begin{equation}\label{xxp-case1} x_i=\!\left\{\!\begin{aligned}
&x'_i, ~~~~\text{for}~i\in[1,d_1-1]\backslash \{\lambda_1,
\lambda_2\};\\
&x'_{i-1}, ~\text{for}~i\in[d_1+1,d_2];\\
&x'_i, ~~~~\text{for}~i\in[d_2+1,n].
\end{aligned}\right.
\end{equation}
Moreover, we have $x_{\lambda_1}\neq x'_{\lambda_1}$ and
$x_{\lambda_2}\neq x'_{\lambda_2}$ because of the substitution
error. According to \eqref{xxp-case1}, this case can be
illustrated by Fig. 1.
\begin{figure}[htbp]\label{tp-intv-3-1}
\begin{center}
\includegraphics[height=1.1cm]{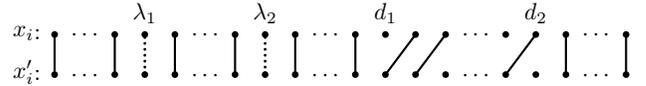}
\end{center}
\caption{Illustration of Case (i): The bits (symbols) of each
sequence is denoted by a row of black dots, where each column
corresponds to the two symbols at the same position in the
respective sequences. Each pair of bits connected by a solid
segment are of equal value, while those connected by a dashed
segment have different values because of the substitution error.}
\end{figure}

We can use \eqref{xxp-case1} or Fig. 1 to simplify $u_i$ for each
$i\in[n]~($In fact, Fig. 1 is more intuitive than
\eqref{xxp-case1}.$)$ as follows.

First, we simplify
$u_i=\sum_{\ell=i}^nx_\ell-\sum_{\ell=i}^nx'_\ell$ for
$i\in[1,\lambda_1]$. From Fig. 1 we can see that all terms in
$\sum_{\ell=i}^nx_\ell$ can be cancelled by their corresponding
terms in $\sum_{\ell=i}^nx'_\ell$ except for $x_{\lambda_1},
x_{\lambda_2}$ and $x_{d_1}$, and all terms in
$\sum_{\ell=i}^nx'_\ell$ can be cancelled except for
$x'_{\lambda_1}, x'_{\lambda_2}$ and $x'_{d_2}$, so we have
\begin{align}\label{case1-u1-0}
u_i&=\sum_{\ell=i}^nx_\ell-\sum_{\ell=i}^nx'_\ell\nonumber\\
&=x_{\lambda_1}+x_{\lambda_2}+x_{d_1}-x'_{\lambda_1}-
x'_{\lambda_2}-x'_{d_2}\end{align} In particular, we have
$\mathsf{wt}(\bm x)-\mathsf{wt}(\bm x')
=\sum_{\ell=1}^nx_\ell-\sum_{\ell=1}^nx'_\ell=
u_1=x_{\lambda_1}+x_{\lambda_2}+x_{d_1}-x'_{\lambda_1}-
x'_{\lambda_2}-x'_{d_2}.$ Note that by Remark \ref{rem-wt-xd},
$\mathsf{wt}(\bm x)=\mathsf{wt}(\bm x')$ and $x_{d_1}=x'_{d_2}$.
Therefore, by \eqref{case1-u1-0}, we have
\begin{align}\label{dx-zero-case1}
0&=\mathsf{wt}(\bm x)-\mathsf{wt}(\bm
x')\nonumber\\&=x_{\lambda_1}+x_{\lambda_2}+x_{d_1}-x'_{\lambda_1}-
x'_{\lambda_2}-x'_{d_2}\nonumber\\
&=x_{\lambda_1}+x_{\lambda_2}-x'_{\lambda_1}-
x'_{\lambda_2}\end{align} and
$$u_i=x_{\lambda_1}+x_{\lambda_2}-x'_{\lambda_1}- x'_{\lambda_2}=0,
~\forall\!~i\in[1,\lambda_1].$$

Similarly, from Fig. 1, by cancelling the corresponding equivalent
terms in $\sum_{\ell=i}^nx_\ell$ and $\sum_{\ell=i}^nx'_\ell$, we
can obtain:
\begin{itemize}
\item $u_i=x_{\lambda_2}+x_{d_1}-x'_{\lambda_2}-x'_{d_2}
 =x_{\lambda_2}-x'_{\lambda_2}$ for each $i\in[\lambda_1+1,\lambda_2]$,
 where the second equality holds because $x_{d_1}
 =x'_{d_2}~($according to Remark \ref{rem-wt-xd}$)$.
 \item $u_i=x_{d_1}-x'_{d_2}=0$ for each $i\in[\lambda_2+1,d_1]$.
 \item $u_i=x_{i}-x'_{d_2}$ for each $i\in[d_1+1,d_2]$.
 \item $u_i=0$ for each $i\in[d_2+1,n]$.
\end{itemize}
Collectively, we have
\begin{equation}\label{xxp-case1-1} u_i\!=\!\left\{\!\begin{aligned}
&0, ~~~~~~~~\!~~~~\text{for}~i\in[1,\lambda_1];\\
&x_{\lambda_2}-x'_{\lambda_2},
~\text{for}~i\in[\lambda_1+1,\lambda_2];\\
&0, ~~~~~~~~\!~~~~\text{for}~i\in[\lambda_2+1,d_1];\\
&x_{i}-x'_{d_2}, ~~~\text{for}~i\in[d_1+1,d_2];\\
&0, ~~~~~~~~\!~~~~\text{for}~i\in[d_2+1,n].
\end{aligned}\right.
\end{equation}
Moreover, we have the following claim.

\emph{Claim 1}: Let $p_1=\lambda_2$. Then for each $j\in\{1,2\}$,
either $u_i\geq 0$ for all $i\in[p_{j-1}+1,p_j]$ or $u_i\leq 0$
for all $i\in[p_{j-1}+1,p_j]$, where $p_0=1$ and $p_2=n$.
Moreover, $|u_i|\leq 1$ for all $i\in[n]$.

\begin{proof}[Proof of Claim 1]
For $i\in[1,\lambda_2]$, by \eqref{xxp-case1-1}, we have $u_i=0$
or $u_i=x_{\lambda_2}-x'_{\lambda_2}$. If $x'_{\lambda_2}=0$, then
$u_i\in\{0,1\}$ for all $i\in[\lambda_2+1,n]$; if
$x'_{\lambda_2}=1$, then $u_i\in\{-1,0\}$ for all
$i\in[\lambda_2+1,n]$.

For $i\in[\lambda_2+1,n]$, by \eqref{xxp-case1-1}, we have $u_i=0$
or $u_i=x_{i}-x'_{d_2}$. If $x'_{d_2}=0$, then $u_i\in\{0,1\}$ for
all $i\in[\lambda_2+1,n]$; if $x'_{d_2}=1$, then $u_i\in\{-1,0\}$
for all $i\in[\lambda_2+1,n]$.

Thus, $p_1=\lambda_2$ satisfies the desired property and
$|u_i|\leq 1$ for all $i\in[n]$, which proves Claim 1.
\end{proof}

\vspace{2pt}By \eqref{re-def-f}, for $j=1,2$, we have
\begin{align}\label{xxp-case1-2}|f_j(\bm x)-f_j(\bm
x')|&=\left|\sum_{i=1}^n\left(\sum_{\ell=i}^nx_\ell\right)i^{j-1}-
\sum_{i=1}^n\left(\sum_{\ell=i}^nx'_\ell\right)i^{j-1}\right|\nonumber\\
&=\left|\sum_{i=1}^nu_ii^{j-1}\right|\nonumber\\
&\leq\sum_{i=1}^ni^{j-1}\nonumber\\&<n^j,\end{align} where the
first inequality holds because by Claim 1, $|u_i|\leq 1$ for all
$i\in[n]$. Note that by (C1) and (C2) of Definition
\ref{code-Con}, $f_j(\bm x)\equiv f_j(\bm x')~(\text{mod}~2n^j)$,
so by \eqref{xxp-case1-2}, we have $f_j(\bm x)=f_j(\bm x')$. Thus,
by Claim 1 and Lemma \ref{PT-sgn-eq}, we have $\bm x=\bm x'$.

\begin{exam}
To help the reader to understand the proof, consider an
example with
\begin{align*}\bm x\!~&=1101101000101110,\\
\bm y\!~&=110111100101110,\\
\bm x'&=1001111001011010,\end{align*} where $n=16$. We can check
that $\bm y$ can be obtained from $\bm x$ by deleting $x_{10}=0$
and substituting $x_6=0$ with $y_6=\bar{x}_6=1$, and $\bm y$ can
also be obtained from $\bm x'$ by deleting $x'_{14}=0$ and
substituting $x'_2=0$ with $y_2=\bar{x}'_2=1$. Hence, $\bm y=E(\bm
x,10,6)=E(\bm x',14,2)$, that is, $d_1=10$, $e_1=6$, $d_2=14$ and
$e_2=2$. Since $e_2<e_1$, we take $\lambda_1=e_2=2$ and
$\lambda_2=e_1=6$. For this example, $\bm x$ and $\bm x'$ can be
illustrated by Fig. 2, which is an instance of Fig. 1. It is easy
to check that:
\begin{itemize}
 \item For $i\in[\lambda_1]=\{1,2\}$,
 $u_i=\sum_{\ell=i}^nx_\ell-\sum_{\ell=i}^nx_\ell'=x_{\lambda_1}
 +x_{\lambda_2}+x_{d_1}-x_{\lambda_1}'-x_{\lambda_2}'-x'_{d_2}=0$;
 \vspace{2pt}\item For $i\in[\lambda_1+1,\lambda_2]=\{3,4,5,6\}$,
 $u_i=\sum_{\ell=i}^nx_\ell-\sum_{\ell=i}^nx_\ell'=
 x_{\lambda_2}+x_{d_1}-x_{\lambda_2}'-x'_{d_2}
 =x_{\lambda_2}-x_{\lambda_2}'=-1$;
 \vspace{2pt}\item For $i\in[\lambda_2+1,d_1]=\{7,8,9,10\}$,
 $u_i=\sum_{\ell=i}^nx_\ell-\sum_{\ell=i}^nx_\ell'=
 x_{d_1}-x_{d_2}'=0$;
 \vspace{2pt}\item For $i\in[d_1+1,d_2]=\{11,12,13,14\}$,
 $u_i=\sum_{\ell=i}^nx_\ell-\sum_{\ell=i}^nx_\ell'
 =x_{i}-x'_{d_2}=x_i\in\{0,1\}$;
 \vspace{2pt}\item For $i\in[d_2+1,n]=\{15,16\}$,
 $u_i=\sum_{\ell=i}^nx_\ell-\sum_{\ell=i}^nx_\ell'
 =0$;
\end{itemize}
In summary, we have \begin{align*} &(u_1,u_2,\cdots\!,u_n)\\
&=(0,0,-1,-1,-1,-1,0,0,0,0,1,0,1,1,0,0).\end{align*} We can see
that $u_i\leq 0$ for all $i\in[1,\lambda_2]=\{1,2,\cdots,6\}$, and
$u_i\geq 0$ for all $i\in[\lambda_2+1,n]=\{7,8,\cdots,16\}$.
\end{exam}

\begin{figure}[htbp]\label{tp-eg-3-1}
\begin{center}
\includegraphics[height=1.05cm]{fig-exm.1}
\end{center}
\caption{An example of Case (i).}
\end{figure}

\subsection{Proof of $\bm{x}=\bm{x}'$ for Case (ii)}

For this case, we have $e_1<d_1\leq e_2<d_{2}$ or $e_2<d_1<e_1\leq
d_{2}$. If $e_1<d_1\leq e_2<d_{2}$, let $\lambda_1=e_1$ and
$\lambda_2=e_2+1$; If $e_2<d_1<e_1\leq d_{2}$, let $\lambda_1=e_2$
and $\lambda_2=e_1$. Then for both cases, we always have
$\lambda_1<d_1<\lambda_2\leq d_2$. Since
$E(\bm{x},d_1,e_1)=E(\bm{x}',d_2,e_2)$, analogous to
\eqref{xxp-case1}, we can obtain
\begin{equation}\label{xxp-case2} x_i=\!\left\{\!\begin{aligned}
&x'_i, ~~~~\text{for}~i\in[1,d_1-1]\backslash \{\lambda_1\},\\
&x'_{i-1}, ~\text{for}~i\in[d_1+1,d_2]\backslash \{\lambda_2\},\\
&x'_i, ~~~~\text{for}~i\in[d_2+1,n].
\end{aligned}\right.
\end{equation}
Moreover, we have $x_{\lambda_1}\neq x'_{\lambda_1}$ and
$x_{\lambda_2}\neq x'_{\lambda_2-1}$ because of the substitution
error. According to \eqref{xxp-case2}, this case can be
illustrated by Fig. 3.
\begin{figure}[htbp]\label{tp-intv-3-2}
\begin{center}
\includegraphics[height=1.1cm]{case.2}
\end{center}
\caption{Illustration of Case (ii).}
\end{figure}

By Remark \ref{rem-wt-xd}, we have $\mathsf{wt}(\bm
x)=\mathsf{wt}(\bm x')$ and $x_{d_1}=x'_{d_2}$. Then by
\eqref{xxp-case2} or Fig. 3, and through a cancelling process
similar to Case (i), we can obtain $0=\mathsf{wt}(\bm
x)-\mathsf{wt}(\bm
x')=\sum_{\ell=1}^nx_\ell-\sum_{\ell=1}^nx_\ell'
=x_{\lambda_1}+x_{\lambda_2}+x_{d_1}-x'_{\lambda_1}-
x'_{\lambda_2-1}-x'_{d_2}=
x_{\lambda_1}+x_{\lambda_2}-x'_{\lambda_1}- x'_{\lambda_2-1}$ and
\begin{equation}\label{xxp-case2-1} u_i\!=\!\left\{\!\begin{aligned}
&0, ~~~~~~~~\!~~~~\text{for}~i\!\in\![1,\lambda_1];\\
&x_{\lambda_2}-x'_{\lambda_2-1},
~\text{for}~i\in[\lambda_1+1,d_1];\\
&x_{i}+x_{\lambda_2}-x'_{\lambda_2-1}-x'_{d_2},
~\text{for}~i\in[d_1+1,\lambda_2-1];\\
&x_{i}-x'_{d_2}, ~~~\text{for}~i\in[\lambda_2,d_2];\\
&0, ~~~~~~~~\!~~~~\text{for}~i\in[d_2+1,n].
\end{aligned}\right.
\end{equation}
Moreover, we have the following Claim.

\emph{Claim 2}: Let $p_1=\lambda_2-1$. Then for each
$j\in\{1,2\}$, either $u_i\geq 0$ for all $i\in[p_{j-1}+1,p_j]$ or
$u_i\leq 0$ for all $i\in[p_{j-1}+1,p_j]$, where $p_0=1$ and
$p_2=n$. Moreover, we have $|u_i|\leq 2$ for all $i\in[n]$.

\begin{proof}[Proof of Claim 2]
First consider $i\in[\lambda_2,n]$. By \eqref{xxp-case2-1},
$u_i=0$ or $u_i=x_{i}-x'_{d_2}$. Clearly, if $x'_{d_2}=0$, then
$u_i\geq 0$ for all $i\in[\lambda_2,n]$; if $x'_{d_2}=1$, then
$u_i\leq 0$ for all $i\in[\lambda_2,n]$.

Now, consider $i\in[1,\lambda_2-1]$. Note that $x_{\lambda_2}\neq
x'_{\lambda_2-1}$. Then we have
$x_{\lambda_2}-x'_{\lambda_2-1}\in\{-1,1\}$. We need to consider
the following two subcases.

Case (ii.1): $x_{\lambda_2}-x'_{\lambda_2-1}=1$.

By \eqref{xxp-case2-1}, we have
\begin{equation*} u_i\!=\!\left\{\!\begin{aligned}
&0, ~~~~~~~~\text{for}~i\!\in\![1,\lambda_1];\\
&1, ~~~~~~~~\text{for}~i\in[\lambda_1+1,d_1];\\
&x_{i}+1-x'_{d_2}, ~\text{for}~i\in[d_1+1,\lambda_2-1].
\end{aligned}\right.
\end{equation*}
Note that $x_{i}\geq 0$ and $1-x'_{d_2}\geq 0~($because
$x'_{d_2}\in\{0,1\})$. Then $u_i\geq 0$ for all
$i\in[1,\lambda_2-1]$.

Case (ii.2): $x_{\lambda_2}-x'_{\lambda_2-1}=-1$.

By \eqref{xxp-case2-1}, we have
\begin{equation*} u_i\!=\!\left\{\!\begin{aligned}
&0, ~~~~~~~~\text{for}~i\!\in\![1,\lambda_1];\\
&-1, ~~~~~~~\text{for}~i\in[\lambda_1+1,d_1];\\
&x_{i}-1-x'_{d_2}, ~\text{for}~i\in[d_1+1,\lambda_2-1].
\end{aligned}\right.
\end{equation*}
Note that $x_{i}-1\leq 0~($because $x'_{d_2}\in\{0,1\})$ and
$-x'_{d_2}\leq 0$. Then $u_i\leq 0$ for all $i\in[1,\lambda_2-1]$.

Thus, $p_1=\lambda_2-1$ satisfies the desired property.

Finally, note that $|x_{\lambda_2}-x'_{\lambda_2-1}|\leq 1$ and
$|x_{i}-x'_{d_2}|\leq 1$. Then it is easy to see from
\eqref{xxp-case2-1} that $|u_i|\leq 2$ for all $i\in[n]$, which
proves Claim 2.
\end{proof}

\vspace{2pt} Similar to Case (i), by \eqref{re-def-f} and Claim 2,
for $j=1,2$, we have
\begin{align*}|f_j(\bm x)-f_j(\bm
x')|\leq\sum_{i=1}^n|u_i|i^{j-1}\leq\sum_{i=1}^n2i^{j-1}<2n^j.
\end{align*}
On the other hand, by (C1) and (C2) of Definition \ref{code-Con},
we have $f_j(\bm x)\equiv f_j(\bm x')~(\text{mod}~2n^j)$, so
$f_j(\bm x)=f_j(\bm x')$. Then by Claim 2 and Lemma
\ref{PT-sgn-eq}, we have $\bm x=\bm x'$.

\begin{exam}
Consider an example with
\begin{align*}
\bm x\!~&=1001011101001110,\\
\bm y\!~&=110111101001110,\\
  \bm x'&=1101111000011010,\end{align*} where $n=16$. We can check
that $\bm y$ can be obtained from $\bm x$ by deleting $x_{5}=0$
and substituting $x_{2}=0$ with $y_{2}=\bar{x}_{2}=1$, and $\bm y$
can also be obtained from $\bm x'$ by deleting $x'_{14}=0$ and
substituting $x'_{9}=0$ with $y_9=\bar{x}'_9=1$. Hence, we have
$\bm y=E(\bm x,5,2)=E(\bm x',14,9)$, that is, $d_1=5$, $e_1=2$,
$d_2=14$ and $e_2=9$. Since $e_1<d_1<e_2<d_2$, we take
$\lambda_1=e_1=2$ and $\lambda_2=e_2+1=10$. For this example, $\bm
x$ and $\bm x'$ can be illustrated by Fig. 4, which is an instance
of Fig. 3. It is easy to check that:
\begin{itemize}
 \item For $i\in[\lambda_1]=\{1,2\}$,
 $u_i=\sum_{\ell=i}^nx_\ell-\sum_{\ell=i}^nx_\ell'=x_{\lambda_1}
 +x_{\lambda_2}+x_{d_1}-x_{\lambda_1}'-x_{\lambda_2}'-x'_{d_2}=0$;
 \vspace{2pt}\item For $i\in[\lambda_1+1,d_1]=\{3,4,5\}$,
 $u_i=\sum_{\ell=i}^nx_\ell-\sum_{\ell=i}^nx_\ell'=
 x_{d_1}\!+\!x_{\lambda_2}\!-\!x_{\lambda_2-1}'\!-\!x'_{d_2}
 \!=\!x_{\lambda_2}\!-\!x'_{\lambda_2-1}\!=\!1$;
 \vspace{2pt}\item For $i\in[d_1+1,\lambda_2-1]=\{6,7,8,9\}$,
 $u_i=\sum_{\ell=i}^nx_\ell-\sum_{\ell=i}^nx_\ell'=
 x_i+x_{\lambda_2}-x'_{\lambda_2-1}-x_{d_2}'=x_i+1\in\{1,2\}$;
 \vspace{2pt}\item For $i\in[\lambda_2,d_2]=\{10,11,12,13,14\}$,
 $u_i=\sum_{\ell=i}^nx_\ell-\sum_{\ell=i}^nx_\ell'
 =x_{i}-x'_{d_2}=x_{i}\in\{0,1\}$;
 \vspace{2pt}\item For $i\in[d_2+1,n]=\{15,16\}$,
 $u_i=\sum_{\ell=i}^nx_\ell-\sum_{\ell=i}^nx_\ell'=0$;
\end{itemize}
In summary, we have \begin{align*} &(u_1,u_2,\cdots\!,u_n)\\
&=(0,0,1,1,1,2,2,2,1,1,0,0,1,1,0,0).\end{align*} We can see that
$u_i\geq 0$ for all $i\in[1,\lambda_2-1]=\{1,2,\cdots,9\}$, and
$u_i\geq 0$ for all $i\in[\lambda_2,n]=\{10,11,\cdots,16\}$. Note
that in this example, we have $u_i\geq 0$ for all $i\in[n]$, which
is stronger than Claim 2. However, this is not the case in
general.
\end{exam}

\begin{figure}[htbp]\label{tp-eg-3-2}
\begin{center}
\includegraphics[height=1.05cm]{fig-exm.2}
\end{center}
\caption{An example of Case (ii).}
\end{figure}

\subsection{Proof of $\bm{x}=\bm{x}'$ for Case (iii)}

For this case, we have $e_1<d_1\leq d_2<e_{2}$ or $e_2<d_1\leq
d_{2}<e_1$. Let $\lambda_1=\min\{e_1,e_2\}$ and
$\lambda_2=\max\{e_1,e_2\}$. Then we have $\lambda_1<d_1\leq
d_2<\lambda_2$. Since $E(\bm{x},d_1,e_1)=E(\bm{x}',d_2,e_2)$,
analogous to \eqref{xxp-case1}, we can obtain
\begin{equation}\label{xxp-case3} x_i=\!\left\{\!\begin{aligned}
&x'_i, ~~~~\text{for}~i\in[1,d_1-1]\backslash \{\lambda_1\},\\
&x'_{i-1}, ~\text{for}~i\in[d_1+1,d_2],\\
&x'_i, ~~~~\text{for}~i\in[d_2+1,n]\backslash \{\lambda_2\}.
\end{aligned}\right.
\end{equation}
Moreover, we have $x_{\lambda_1}\neq x'_{\lambda_1}$ and
$x_{\lambda_2}\neq x'_{\lambda_2}$ because of the substitution
error. According to \eqref{xxp-case3}, this case can be
illustrated by Fig. 5.
\begin{figure}[htbp]\label{tp-intv-3-3}
\begin{center}
\includegraphics[height=1.1cm]{case.3}
\end{center}
\caption{Illustration of Case (iii).}
\end{figure}

By Remark \ref{rem-wt-xd}, we have $\mathsf{wt}(\bm
x)=\mathsf{wt}(\bm x')$ and $x_{d_1}=x'_{d_2}$. Then by
\eqref{xxp-case3} or Fig. 5, and through a cancelling process
similar to Case (i), we can obtain $0=\mathsf{wt}(\bm
x)-\mathsf{wt}(\bm
x')=\sum_{\ell=1}^nx_\ell-\sum_{\ell=1}^nx_\ell'
=x_{\lambda_1}+x_{\lambda_2}+x_{d_1}-x'_{\lambda_1}-
x'_{\lambda_2}-x'_{d_2}=
x_{\lambda_1}+x_{\lambda_2}-x'_{\lambda_1}- x'_{\lambda_2}$ and
\begin{equation}\label{xxp-case3-1} u_i\!=\!\left\{\!\begin{aligned}
&0, ~~~~~~~~\!~~~~\text{for}~i\!\in\![1,\lambda_1];\\
&x_{\lambda_2}-x'_{\lambda_2},
~\text{for}~i\in[\lambda_1+1,d_1];\\
&x_{i}+x_{\lambda_2}-x'_{d_2}-x'_{\lambda_2},
~\text{for}~i\in[d_1+1,d_2];\\
&x_{\lambda_2}-x'_{\lambda_2}, ~~~\text{for}~i\in[d_2+1,\lambda_2];\\
&0, ~~~~~~~~\!~~~~\text{for}~i\in[\lambda_2+1,n].
\end{aligned}\right.
\end{equation}
Then we have the following Claim.

\emph{Claim 3}: Either $u_i\geq 0$ for all $i\in[n]$ or $u_i\leq
0$ for all $i\in[n]$. Moreover, $|u_i|\leq 2$ for all $i\in[n]$.

\begin{proof}[Proof of Claim 3]
Note that $x_{\lambda_2}\neq x'_{\lambda_2}$. Then we have
$x_{\lambda_2}-x'_{\lambda_2}\in\{-1,1\}$. To prove Claim 3,
similar to Case (ii), we consider the following two subcases.

Case (iii.1): $x_{\lambda_2}-x'_{\lambda_2}=1$.

By \eqref{xxp-case3-1}, we can obtain
\begin{equation*} u_i\!=\!\left\{\!\begin{aligned}
&0, ~~~~~~~\text{for}~i\!\in\![1,\lambda_1];\\
&1, ~~~~~~~\text{for}~i\in[\lambda_1+1,d_1];\\
&x_{i}-x'_{d_2}+1,
~\text{for}~i\in[d_1+1,\lambda_2-1];\\
&1, ~~~~~~\text{for}~i\in[\lambda_2,d_2];\\
&0, ~~~~~~\text{for}~i\in[d_2+1,n].
\end{aligned}\right.
\end{equation*}
Note that $1-x'_{d_2}\geq 0$ and $x_i\geq 0$ for all $i\in[n]$.
Then $u_i\geq 0$ for all $i\in[n]$.

Case (iii.2): $x_{\lambda_2}-x'_{\lambda_2}=-1$.

By \eqref{xxp-case3-1}, we can obtain
\begin{equation*} u_i\!=\!\left\{\!\begin{aligned}
&0, ~~~~~~~\text{for}~i\!\in\![1,\lambda_1];\\
&-1, ~~~~~~~~~\text{for}~i\in[\lambda_1+1,d_1];\\
&x_{i}-x'_{d_2}-1,
~\text{for}~i\in[d_1+1,\lambda_2-1];\\
&-1, ~~~~~~\text{for}~i\in[\lambda_2,d_2];\\
&0, ~~~~~~~\text{for}~i\in[d_2+1,n].
\end{aligned}\right.
\end{equation*}
Since $-x'_{d_2}\leq 0$ and $x_{i}-1\leq 0$ for all $i\in[n]$,
then $u_i\leq 0$ for all $i\in[n]$.

Thus, either $u_i\geq 0$ for all $i\in[n]$ or $u_i\leq 0$ for all
$i\in[n]$.

Note that $|x_{\lambda_2}-x'_{\lambda_2}|\leq 1$ and
$|x_{i}-x'_{d_2}|\leq 1$. It is easy to see from
\eqref{xxp-case3-1} that $|u_i|\leq 2$ for all $i\in[n]$, which
proves Claim 3.
\end{proof}

\vspace{2pt} Similar to Case (i), by \eqref{re-def-f} and by Claim
3, for $j=1,2$,
\begin{align*}|f_j(\bm x)-f_j(\bm
x')|\leq\sum_{i=1}^n|u_i|i^{j-1}\leq\sum_{i=1}^n2i^{j-1}<2n^j.\end{align*}
On the other hand, by (C1) and (C2) of Definition \ref{code-Con},
we have $f_j(\bm x)\equiv f_j(\bm x')~(\text{mod}~2n^j)$, so
$f_j(\bm x)=f_j(\bm x')$. Then by Claim 3 and Lemma
\ref{PT-sgn-eq}, we have $\bm x=\bm x'$.

\begin{exam}
Consider an example with
\begin{align*}
\bm x\!~&=1001010101001111,\\
\bm y\!~&=100110101001101,\\
  \bm x'&=1101101010001101,\end{align*} where $n=16$. We can check
that $\bm y$ can be obtained from $\bm x$ by deleting $x_{5}=0$
and substituting $x_{15}=1$ with $y_{14}=\bar{x}_{15}=0$, and $\bm
y$ can also be obtained from $\bm x'$ by deleting $x'_{10}=0$ and
substituting $x'_{2}=1$ with $y_2=\bar{x}'_2=0$. Hence, $\bm
y=E(\bm x,5,15)=E(\bm x',10,2)$, that is, $d_1=5$, $e_1=15$,
$d_2=10$ and $e_2=2$. Since $e_2<e_1$, we take $\lambda_1=e_2=2$
and $\lambda_2=e_1=15$. For this example, $\bm x$ and $\bm x'$ can
be illustrated by Fig. 6, which is an instance of Fig. 5. It is
easy to check that:
\begin{itemize}
 \item For $i\in[\lambda_1]=\{1,2\}$,
 $u_i=\sum_{\ell=i}^nx_\ell-\sum_{\ell=i}^nx_\ell'=x_{\lambda_1}
 +x_{\lambda_2}+x_{d_1}-x_{\lambda_1}'-x_{\lambda_2}'-x'_{d_2}=0$;
 \vspace{2pt}\item For $i\in[\lambda_1+1,d_1]=\{3,4,5\}$,
 $u_i=\sum_{\ell=i}^nx_\ell-\sum_{\ell=i}^nx_\ell'=
 x_{d_1}\!+\!x_{\lambda_2}\!-\!x'_{d_2}-\!x_{\lambda_2}'\!
 \!=\!x_{\lambda_2}\!-\!x'_{\lambda_2}\!=\!1$;
 \vspace{2pt}\item For $i\in[d_1+1,d_2]=\{6,7,8,9,10\}$,
 $u_i=\sum_{\ell=i}^nx_\ell-\sum_{\ell=i}^nx_\ell'=
 x_i+x_{\lambda_2}-x_{d_2}'-x'_{\lambda_2}=x_i+1\in\{1,2\}$;
 \vspace{2pt}\item For $i\in[d_2+1,\lambda_2]=\{11,12,13,14,15\}$,
 $u_i=\sum_{\ell=i}^nx_\ell-\sum_{\ell=i}^nx_\ell'
 =x_{\lambda_2}-x'_{\lambda_2}=1$;
 \vspace{2pt}\item For $i\in[\lambda_2+1,n]=\{16\}$,
 $u_i=\sum_{\ell=i}^nx_\ell-\sum_{\ell=i}^nx_\ell'=0$;
\end{itemize}
In summary, we have \begin{align*} &(u_1,u_2,\cdots\!,u_n)\\
&=(0,0,1,1,1,2,1,2,1,2,1,1,1,1,1,0).\end{align*} We can see that
$u_i\geq 0$ for all $i\in[n]=\{1,2,\cdots,16\}$.
\end{exam}

\begin{figure}[htbp]\label{tp-eg-3-3}
\begin{center}
\includegraphics[height=1.05cm]{fig-exm.3}
\end{center}
\caption{An example of Case (iii).}
\end{figure}

\subsection{Proof of $\bm{x}=\bm{x}'$ for Case (v)}

For this case, we have $d_1<e_1\leq d_2<e_{2}$ or $d_1\leq
e_2<d_{2}<e_1$. If $d_1<e_1\leq d_2<e_{2}$, let $\lambda_1=e_1$
and $\lambda_2=e_2$; If $d_1\leq e_2<d_{2}<e_1$, let
$\lambda_1=e_2+1$ and $\lambda_2=e_1$. Then for both cases, we
always have $d_1<\lambda_1\leq d_2<\lambda_2$. Since
$E(\bm{x},d_1,e_1)=E(\bm{x}',d_2,e_2)$, analogous to
\eqref{xxp-case1}, we can obtain
\begin{equation}\label{xxp-case5} x_i=\!\left\{\!\begin{aligned}
&x'_i, ~~~~\text{for}~i\in[1,d_1-1],\\
&x'_{i-1}, ~\text{for}~i\in[d_1+1,d_2]\backslash \{\lambda_1\},\\
&x'_i, ~~~~\text{for}~i\in[d_2+1,n]\backslash \{\lambda_2\}.
\end{aligned}\right.
\end{equation}
Moreover, we have $x_{\lambda_1}\neq x'_{\lambda_1-1}$ and
$x_{\lambda_2}\neq x'_{\lambda_2}$ because of the substitution
error. According to \eqref{xxp-case5}, this case can be
illustrated by Fig. 7.
\begin{figure}[htbp]\label{tp-intv-3-5}
\begin{center}
\includegraphics[height=1.1cm]{case.5}
\end{center}
\caption{Illustration of Case (v).}
\end{figure}

By Remark \ref{rem-wt-xd}, we have $\mathsf{wt}(\bm
x)=\mathsf{wt}(\bm x')$ and $x_{d_1}=x'_{d_2}$. Then by
\eqref{xxp-case5} or Fig. 7, and through a cancelling process
similar to Case (i), we can obtain $0=\mathsf{wt}(\bm
x)-\mathsf{wt}(\bm
x')=\sum_{\ell=1}^nx_\ell-\sum_{\ell=1}^nx_\ell'
=x_{\lambda_1}+x_{\lambda_2}+x_{d_1}-x'_{\lambda_1-1}-
x'_{\lambda_2}-x'_{d_2}=
x_{\lambda_1}+x_{\lambda_2}-x'_{\lambda_1-1}- x'_{\lambda_2}$ and
\begin{equation}\label{xxp-case5-1} u_i\!=\!\left\{\!\begin{aligned}
&0, ~~~~~~~~\!~~~~\text{for}~i\!\in\![1,d_1];\\
&x_{i}-x'_{d_2}, ~\text{for}~i\in[d_1+1,\lambda_1-1];\\
&x_{i}+x_{\lambda_2}-x'_{d_2}-x'_{\lambda_2},
~\text{for}~i\in[\lambda_1,d_2];\\
&x_{\lambda_2}-x'_{\lambda_2}, ~~~\text{for}~i\in[d_2+1,\lambda_2];\\
&0, ~~~~~~~~\!~~~~\text{for}~i\in[\lambda_2+1,n].
\end{aligned}\right.
\end{equation}
Then we have the following Claim.

\emph{Claim 4}: Let $p_1=\lambda_1-1$. Then for each
$j\in\{1,2\}$, either $u_i\geq 0$ for all $i\in[p_{j-1}+1,p_j]$ or
$u_i\leq 0$ for all $i\in[p_{j-1}+1,p_j]$, where $p_0=1$ and
$p_2=n$. Moreover, we have $|u_i|\leq 2$ for all $i\in[n]$.

\begin{proof}[Proof of Claim 4]
For $i\in[1,\lambda_1-1]$, by \eqref{xxp-case5-1}, $u_i=0$ or
$u_i=x_{i}-x'_{d_2}$. Clearly, if $x'_{d_2}=0$, then $u_i\geq 0$
for all $i\in[1,\lambda_1-1]$; if $x'_{d_2}=1$, then $u_i\leq 0$
for all $i\in[1,\lambda_1-1]$.

For $i\in[\lambda_1, n]$, similar to Case (ii), we need to
consider the following two subcases.

Case (v.1): $x_{\lambda_2}-x'_{\lambda_2}=1$.

By \eqref{xxp-case5-1}, we have
\begin{equation*} u_i\!=\!\left\{\!\begin{aligned}
&x_{i}-x'_{d_2}+1,
~\text{for}~i\in[\lambda_1,d_2];\\
&1, ~~~~\text{for}~i\in[d_2+1,\lambda_2];\\
&0, ~~~~\text{for}~i\in[\lambda_2+1,n].
\end{aligned}\right.
\end{equation*}
Note that $1-x'_{d_2}\geq 0$ and $x_{i}\geq 0$ for all $i\in[n]$.
Then $u_i\geq 0$ for all $i\in[\lambda_1, n]$.

Case (v.2): $x_{\lambda_2}-x'_{\lambda_2}=-1$.

By \eqref{xxp-case5-1}, we have
\begin{equation*} u_i\!=\!\left\{\!\begin{aligned}
&x_{i}-x'_{d_2}-1,
~\text{for}~i\in[\lambda_1,d_2];\\
&-1, ~~~~\text{for}~i\in[d_2+1,\lambda_2];\\
&0, ~~~~\text{for}~i\in[\lambda_2+1,n].
\end{aligned}\right.
\end{equation*}
Since $-x'_{d_2}\leq 0$ and $x_{i}-1\leq 0$ for all $i\in[n]$,
then $u_i\leq 0$ for all $i\in[\lambda_1, n]$.

Thus, $p_1=\lambda_1-1$ satisfies the desired property.

Note that $|x_{\lambda_2}-x'_{\lambda_2}|\leq 1$ and
$|x_{i}-x'_{d_2}|\leq 1$. It is easy to see from
\eqref{xxp-case5-1} that $|u_i|\leq 2$ for all $i\in[n]$, which
proves Claim 4.
\end{proof}

\vspace{2pt} Similar to Case (i), by \eqref{re-def-f} and by Claim
4, for $j=1,2$,
\begin{align*}|f_j(\bm x)-f_j(\bm
x')|\leq\sum_{i=1}^n|u_i|i^{j-1}\leq\sum_{i=1}^n2i^{j-1}<2n^j.\end{align*}
On the other hand, by (C1) and (C2) of Definition \ref{code-Con},
we have $f_j(\bm x)\equiv f_j(\bm x')~(\text{mod}~2n^j)$, so
$f_j(\bm x)=f_j(\bm x')$. Then by Claim 4 and Lemma
\ref{PT-sgn-eq}, we have $\bm x=\bm x'$.

\subsection{Proof of $\bm{x}=\bm{x}'$ for Case (vi)}

For this case, we have $d_1\leq d_2<e_1$ and $d_1\leq d_{2}<e_2$.
Similar to Case (i), if $e_1=e_2$, then $\bm x=\bm x'$. Therefore,
we assume $e_1\neq e_2$. Let $\lambda_1=\min\{e_1,e_2\}$ and
$\lambda_2=\max\{e_1,e_2\}$. Then we have $d_1\leq
d_2<\lambda_1<\lambda_2$. Since
$E(\bm{x},d_1,e_1)=E(\bm{x}',d_2,e_2)$, analogous to
\eqref{xxp-case1}, we can obtain
\begin{equation}\label{xxp-case6} x_i=\!\left\{\!\begin{aligned}
&x'_i, ~~~~\text{for}~i\in[1,d_1-1],\\
&x'_{i-1}, ~\text{for}~i\in[d_1+1,d_2],\\
&x'_i, ~~~~\text{for}~i\in[d_2+1,n]\backslash \{\lambda_1,
\lambda_2\}.
\end{aligned}\right.
\end{equation}
Moreover, we have $x_{\lambda_1}\neq x'_{\lambda_1}$ and
$x_{\lambda_2}\neq x'_{\lambda_2}$ because of the substitution
error. According to \eqref{xxp-case6}, this case can be
illustrated by Fig. 8.
\begin{figure}[htbp]\label{tp-intv-3-6}
\begin{center}
\includegraphics[height=1.1cm]{case.6}
\end{center}
\caption{Illustration of Case (vi).}
\end{figure}

By Remark \ref{rem-wt-xd}, we have $\mathsf{wt}(\bm
x)=\mathsf{wt}(\bm x')$ and $x_{d_1}=x'_{d_2}$. Then by
\eqref{xxp-case6} or Fig. 8, and through a cancelling process
similar to Case (i), we can obtain $0=\mathsf{wt}(\bm
x)-\mathsf{wt}(\bm
x')=\sum_{\ell=1}^nx_\ell-\sum_{\ell=1}^nx_\ell'
=x_{\lambda_1}+x_{\lambda_2}+x_{d_1}-x'_{\lambda_1}-
x'_{\lambda_2}-x'_{d_2}=
x_{\lambda_1}+x_{\lambda_2}-x'_{\lambda_1}- x'_{\lambda_2}$ and
\begin{equation}\label{xxp-case6-1} u_i\!=\!\left\{\!\begin{aligned}
&0, ~~~~~~~~\!~~~~\text{for}~i\!\in\![1,d_1];\\
&x_{i}-x'_{d_2}, ~\text{for}~i\in[d_1+1,d_2];\\
&0, ~~~~~~~\text{for}~i\in[d_2+1,\lambda_1];\\
&x_{\lambda_2}-x'_{\lambda_2},
~~~\text{for}~i\in[\lambda_1+1,\lambda_2];\\
&0, ~~~~~~~~\!~~~~\text{for}~i\in[\lambda_2+1,n].
\end{aligned}\right.
\end{equation}
Then we have the following Claim.

\emph{Claim 5}: Let $p_1=d_2$. Then for each $j\in\{1,2\}$, either
$u_i\geq 0$ for all $i\in[p_{j-1}+1,p_j]$ or $u_i\leq 0$ for all
$i\in[p_{j-1}+1,p_j]$, where $p_0=1$ and $p_2=n$. Moreover,
$|u_i|\leq 1$ for all $i\in[n]$.

\begin{proof}[Proof of Claim 5]
For $i\in[1, d_2]$, by \eqref{xxp-case6-1}, we have $u_i=0$ or
$u_i=x_{i}-x'_{d_2}$. If $x'_{d_2}=0$, then $u_i\geq 0$ for all
$i\in[1, d_2]$; if $x'_{d_2}=1$, then $u_i\leq 0$ for all $i\in[1,
d_2]$.

For $i\in[d_2+1, n]$, by \eqref{xxp-case6-1}, we have $u_i=0$ or
$u_i=x_{\lambda_2}-x'_{\lambda_2}$. If $x'_{\lambda_2}=0$, then
$u_i\geq 0$ for all $i\in[d_2+1, n]$; if $x'_{\lambda_2}=1$, then
$u_i\leq 0$ for all $i\in[d_2+1, n]$.

Thus, $p_1=d_2$ satisfies the desired property.

Note that $|x_{\lambda_2}-x'_{\lambda_2}|\leq 1$ and
$|x_{i}-x'_{d_2}|\leq 1$. It is easy to see from
\eqref{xxp-case6-1} that $|u_i|\leq 1$ for all $i\in[n]$, which
proves Claim 5.
\end{proof}

\vspace{2pt} Similar to Case (i), by \eqref{re-def-f} and by Claim
5, for $j=1,2$,
\begin{align*}|f_j(\bm x)-f_j(\bm
x')|\leq\sum_{i=1}^n|u_i|i^{j-1}\leq\sum_{i=1}^ni^{j-1}<n^j.\end{align*}
On the other hand, by (C1) and (C2) of Definition \ref{code-Con},
we have $f_j(\bm x)\equiv f_j(\bm x')~(\text{mod}~2n^j)$, so
$f_j(\bm x)=f_j(\bm x')$. Then by Claim 5 and Lemma
\ref{PT-sgn-eq}, we have $\bm x=\bm x'$.




\begin{thebibliography}{1} 

\bibitem{Levenshtein65}
V. I. Levenshtein, ``Binary codes capable of correcting deletions,
insertions and reversals (in Russian),'' \emph{Doklady Akademii
Nauk SSR}, vol. 163, no. 4, pp. 845-848, 1965.

\bibitem{Varshamov65}
R. R. Varshamov and G. M. Tenengolts, ``Codes which correct single
asymmetric errors (in Russian),'' \emph{Automatika i
Telemkhanika}, vol. 161, no. 3, pp. 288-292, 1965.

\bibitem{Cai19}
K. Cai, Y. M. Chee, R. Gabrys, H. M. Kiah, and T. T. Nguyen,
``Correcting a Single Indel/Edit for DNA-Based Data Storage:
Linear-Time Encoders and Order-Optimality,'' \emph{IEEE Trans.
Inform. Theory}, vol. 67, no. 6, pp. 3438-3451, June 2021.

\bibitem{Gabrys-21}
R. Gabrys, V. Guruswami, J. Ribeiro, and  K. Wu, ``Beyond
Single-Deletion Correcting Codes: Substitutions and
Transpositions,'' 2021, Available online at:
https://arxiv.org/abs/2112.09971

\bibitem{Helberg02}
A. S. Helberg and H. C. Ferreira, ``On multiple insertion/deletion
correcting codes,'' \emph{IEEE Trans. Inform. Theory}, vol. 48,
no. 1, pp. 305-308, Jan. 2002.

\bibitem{Brakensiek18}
J. Brakensiek, V. Guruswami, and S. Zbarsky, ``Efficient
low-redundancy codes for correcting multiple deletions,''
\emph{IEEE Trans. Inform. Theory}, vol. 64, no. 5, pp. 3403-3410,
2018.

\bibitem{Gabrys19}
R. Gabrys and F. Sala, ``Codes correcting two deletions,''
\emph{IEEE Trans. Inform. Theory}, vol. 65, no. 2, pp. 965-974,
Feb 2019.

\bibitem{Sima19-1}
J. Sima, N. Raviv, and J. Bruck, ``Two deletion correcting codes
from indicator vectors,'' \emph{IEEE Trans. Inform. Theory}, vol.
66, no. 4, pp. 2375-2391, April 2020.

\bibitem{Sima19-2}
J. Sima and J. Bruck, ``Optimal $k$-Deletion Correcting Codes,''
in \emph{Proc. ISIT}, 2019.

\bibitem{Sima19-2-1}
J. Sima and J. Bruck, ``On Optimal $k$-Deletion Correcting
Codes,'' \emph{IEEE Trans. Inform. Theory}, vol. 67, no. 6, pp.
3360-3375, June 2021.

\bibitem{Sima20}
J. Sima, R. Gabrys, and J. Bruck, ``Optimal Systematic
$t$-Deletion Correcting Codes,'' in \emph{Proc. ISIT}, 2020.

\bibitem{Gur2020}
V. Guruswami and Johan H\aa stad, ``Explicit two-deletion codes
with redundancy matching the existential bound,'' \emph{IEEE
Trans. Inform. Theory}, vol. 67, no. 10, pp. 6384-6393, October
2021.

\bibitem{Smagloy20}
I. Smagloy, L. Welter, A. Wachter-Zeh, and E. Yaakobi,
``Single-Deletion Single-Substitution Correcting Codes,'' in
\emph{Proc. ISIT}, 2020.

\bibitem{SNCH-20}
W. Song, N. Polyanskii, K. Cai, and X. He, ``Systematic
Single-Deletion Multiple-Substitution Correcting Codes,'' 2020,
Available online at: https://arxiv.org/abs/2006.11516

\bibitem{SNCH-21}
W. Song, N. Polyanskii, K. Cai, and X. He, ``On multiple-deletion
multiple-substitution correcting codes,'' in \emph{Proc. ISIT},
2021.

\bibitem{Guruswami17}
V. Guruswami and C. Wang, ``Deletion codes in the high-noise and
highrate regimes,'' \emph{IEEE Trans. Inform. Theory}, vol. 63,
no. 4, pp. 1961-1970, April 2017.

\bibitem{Wachter18}
A. Wachter-Zeh, ``List decoding of insertions and deletions,''
\emph{IEEE Trans. Inform. Theory}, vol. 64, no. 9, pp. 6297-6304,
September 2018.

\bibitem{Hayashi18}
T. Hayashi and K. Yasunaga, ``On the list decodability of
insertions and deletions,'' \emph{IEEE Trans. Inform. Theory},
vol. 66, no. 9, pp. 5335-5343, September 2020.

\bibitem{Haeupler18}
B. Haeupler, A. Shahrasbi, and M. Sudan, ``Synchronization
strings: List decoding for insertions and deletions,'' in
\emph{Proc. 45th Int. Colloq. Automata, Lang. Program. (ICALP)},
2018.

\bibitem{Liu21}
S. Liu, I. Tjuawinata, and C. Xing, ``Efficiently list-decodable
insertion and deletion codes via concatenation,'' \emph{IEEE
Trans. Inform. Theory}, vol. 67, no. 9, pp. 5778-5790, September
2021.

\end{thebibliography}
\end{document}